\documentclass[10pt,letterpaper]{article}

\usepackage[hang,perpage]{footmisc}

\usepackage{amsmath}
\usepackage{amsthm}
\usepackage{amssymb}
\usepackage{stmaryrd}
\usepackage{graphicx}
\usepackage{latexsym}
\usepackage{times}
\usepackage[usenames]{color}
\usepackage{braket}
\usepackage{textcomp}
\usepackage{makeidx}
\usepackage{tocbibind}
\usepackage[bookmarks=true,backref=true,colorlinks=true,linkcolor=darkolivegreen,
citecolor=darkolivegreen,urlcolor=darkolivegreen]{hyperref}
\usepackage{authblk}

\newcommand{\sgn}{\operatorname{sgn}}

\theoremstyle{definition}
\newtheorem{ass}{Assumption}[section]
\newtheorem{theorem}[ass]{Theorem}
\newtheorem{lemma}[ass]{Lemma}

\newtheorem{corollary}[ass]{Corollary}

\def\indexname{Index of terminology}
\makeindex

\newcommand{\captionfonts}{\footnotesize}
\makeatletter  
\long\def\@makecaption#1#2{%
  \vskip\abovecaptionskip
  \sbox\@tempboxa{{\captionfonts #1: #2}}%
  \ifdim \wd\@tempboxa >\hsize
    {\captionfonts #1: #2\par}
  \else
    \hbox to\hsize{\hfil\box\@tempboxa\hfil}%
  \fi
  \vskip\belowcaptionskip}
\makeatother   
\definecolor{darkolivegreen}{rgb}{0.333333, 0.419608, 0.1843140}
\allowdisplaybreaks

\makeatletter
\def\printnotation{{%
\def\indexname{Index of notation}
\begin{theindex}
\@input{\jobname.ntn}
\end{theindex}
}}
\makeatother

\makeglossary

\begin{document}


\title{Nonuniqueness of Representations of 
Wave Equations in Lorentzian Space-Times}

\author[1,2]{Horst Reinhard Beyer}

\affil[1]{Instituto Tecnol\'ogico Superior de Uruapan, 
Carr. Uruapan-Carapan No. 5555, Col. La Basilia, Uruapan, 
Michoac\'an, M\'exico}

\affil[2]{Theoretical Astrophysics, IAAT, Eberhard Karls 
University of T\"ubingen, T\"ubingen 72076, Germany}

\date{\today}                                     

\maketitle

\begin{abstract}
This brief note wants to bring to attention that the
formulation of physically reasonable initial-boundary
value problems for wave equations in Lorentzian space-times
is not unique, i.e., that there are inequivalent such 
formulations that lead to a different outcome 
of the stability discussion of the solutions. 
For demonstration, the paper uses the case of 
the wave equation on the right 
Rindler wedge in $2$-dimensional 
Minkowski space. The used methods can be generalized to
wave equations on stationary globally hyperbolic space-times 
with horizons in higher dimensions, such as
Schwarzschild and Kerr space-times. 
\end{abstract}

\section{Introduction}
\label{introduction}

The stability discussion of solutions of Einstein's field 
equations usually lead on wave equations in Lorentzian 
space-times, describing perturbations of the metric, 
together with physically boundary conditions.
In a second step, the definition of the precise class 
of the considered solutions is specified, i.e., a data 
space is chosen for the solutions of the equations that 
leads on a well-posed initial-boundary value problem. 
The latter provides the basis for a meaningful discussion of 
the stability of the solutions, i.e., the existence or 
non-existence of exponentially growing solutions. \\

This brief note wants to bring to attention that this
process is not unique, i.e., that there are inequivalent
ways of formulating initial-boundary value problems
for such wave equations and at the same time that the 
choice of the formulation can affect the outcome of 
the stability discussion. One such example is given  
by the treated relatively simple case of the 
wave equation on the right 
Rindler wedge in $2$-dimensional 
Minkowski space.   
  
\section{Restriction of the Wave Equation to the Right Wedge
of $2$-dimensional Minkowski Space-Time}

In the following, we consider the solutions of the wave 
equation on the right wedge, 
\begin{equation*}
R := \{(x^0,x^1) \in {\mathbb{R}}^2 : x^1 > |x^0|\} \, \, ,
\end{equation*}
of $2$-dimensional Minkowski 
space-time $({\mathbb{R}}^2,g)$, where
\begin{equation*}
g = ({\mathbb{R}}^2,dx^0 \otimes dx^0  - 
dx^1 \otimes dx^1) \, \, , 
\end{equation*} 
and $(x^0,x^1) : {\mathbb{R}}^2 \rightarrow {\mathbb{R}}^2$
denotes an inertial coordinate system. \footnote{If not 
otherwise indicated, the symbols 
$x^0,x^1,\tau,\xi$ denote coordinate projections whose 
domains will be obvious from the context. In addition, 
we assume 
the composition 
of maps, which includes addition, multiplication and so forth, 
always to be maximally defined. For instance, the sum of two 
complex-valued maps is defined on the intersection of their domains.}  
For the coordinatization of $R$, we use well-known 
``Rindler coordinates'' $(\tau,\xi) : R \rightarrow 
{\mathbb{R}}^2$ given by 
\begin{equation} \label{rindler}
\tau(x^0,x^1) := 
\frac{1}{2} \, \ln\left(\frac{x^1+x^0}{x^1-x^0}\right)
\, \, , \, \,  
\xi(x^0,x^1) := \frac{1}{2K} \, 
\ln\left(K^2[(x^1)^2 - (x^0)^2]\right)  
\end{equation} 
for all $(x^0,x^1) \in R$. Here $K > 0$ is 
a constant having the dimension $1/\textrm{length}$.
The inverse transformation to (\ref{rindler}) is given by   
\begin{equation*}
(\tau,\xi)^{-1}(\tau,\xi) =  
\frac{1}{K} \, \left(
e^{K\xi} \sinh(\tau), e^{K\xi} \cosh(\tau)
\right)
\end{equation*}
for every $(\tau,\xi) \in {\mathbb{R}}^2$, and 
the restriction $g|_{R}$ of $g$ to $R$ is given by
 
\begin{equation*}
g|_{R} =  e^{2 K\xi} \left(\frac{1}{K^2} \, 
d\tau \otimes d\tau - d\xi \otimes d\xi\right) \, \, .
\end{equation*}

Solutions $u \in C^{2}(R,{\mathbb{C}})$ of the wave equation on 
$R$ satisfy
\begin{equation} \label{waveequation}
\Box u = 
K^2 e^{-2K\xi} \partial^2_{\tau} u - e^{-2K\xi} 
\partial^2_{\xi} u = 0 \, \, ,
\end{equation}
or equivalently 
\begin{align}
\partial^2_{\tau} u & = - \left(
- \frac{1}{K^2} \,  \partial^2_{\xi}
\right) u \, \, , \label{eq:1} \\
e^{-K\xi} \partial^2_{\tau} \left(e^{-K\xi} u\right) & 
= - e^{-K\xi} \left[ e^{-K\xi}  
\left(
- \frac{1}{K^2} \,  \partial^2_{\xi} 
\right) e^{K\xi}\right]  \left(e^{-K\xi} u\right) \, \, .
\label{eq:2} 
\end{align}
Of course, many other forms of (\ref{waveequation})
are possible. In the following, we will consider only
(\ref{eq:1}) and (\ref{eq:2}).\\

As said, both equations 
are equivalent on the level of $C^2$-solutions.
For arbitrarily given data 
$g_1 \in C^2({\mathbb{R}},{\mathbb{R}})$ and $g_2 \in 
C^1({\mathbb{R}},{\mathbb{R}})$, there is a unique solution 
$u \in C^2({\mathbb{R}}^2,{\mathbb{R}})$ to these 
equations such that 
\begin{equation*} 
u(0,\xi) = g_1(\xi) \, \, , \, \, 
\frac{\partial u}{\partial \tau}(0,\xi) = g_2(\xi) 
\end{equation*}
for every $\xi \in {\mathbb{R}}$. Moreover, this solution 
is given by 
\begin{align} \label{sol11wave}
& u(\tau,\xi) =  \frac{1}{2} \, \bigg[ \, 
g_1\left(\xi + \frac{\tau}{K}\right) + g_1\left(\xi -
\frac{\tau}{K}\right) 
+ K \int_{\xi-\frac{\tau}{K}}^{\xi+\frac{\tau}{K}} 
g_2(s) \, ds \bigg]
\end{align}
for all $(\tau,\xi) \in {\mathbb{R}}^2$, and
\begin{align} \label{sol11waveinertial}
u(x^0,x^1) = \, \, & \frac{1}{2} \, \bigg\{ \, 
g_1\left(\frac{1}{K} \ln[K(x^1+x^0)]\right) + g_1
\left(\frac{1}{K} \ln[K(x^1-x^0)]\right) \nonumber 
\\ 
& \quad \, \, \, \, + K \int_{\frac{1}{K} \ln[K(x^1-x^0)]}^{\frac{1}{K} \ln[K(x^1+x^0)]} 
g_2(s) \, ds \bigg\}
\end{align}
for all $(x^0,x^1) \in R$.
\\

On the other hand, (\ref{eq:1}) results 
from (\ref{waveequation}) by solution for the highest 
time derivative. In this, (\ref{waveequation}) 
is divided by the unbounded function $K^2 e^{-2K\xi}$,
and later (\ref{eq:1}) will be treated analogous to the 
wave equation on $2$-dimensional Minkowski space. 
In this step geometrical information is lost. $(R,g|_{R})$
and Minkowski space are both globally hyperbolic, but 
$(R,g|_{R})$ is geodesically incomplete, i.e., there 
are maximal geodesics whose domains are proper subsets 
of ${\mathbb{R}}$, whereas Minkowski space is geodesically 
complete. It needs to be stressed that such loss of geometrical 
information is not particular to wave equations on 
$2$-dimensional Lorentzian 
space-times, but also happens in higher dimensions.
In addition, later (\ref{eq:1}) will be treated using methods
from operator theory, where an unbounded 
function corresponds to an 
unbounded (or ``discontinuous'') 
operator. Therefore, also from an operator theory perspective, 
the ``division'' of (\ref{waveequation}) by $K^2 e^{-2K\xi}$
needs consideration. \\
 
We are going to see that natural functional analytic 
treatments of (\ref{eq:1}), (\ref{eq:2}) lead to inequivalent 
well-posed initial value formulations for (\ref{waveequation}).
In this connection, it needs to be taken into account that 
(\ref{waveequation}) admits $C^2$-solutions of 
stronger than exponential growth in space and time, 
for instance,   
\begin{equation*}
u(\tau,\xi) = \exp\left\{\frac{\alpha}{K}\exp\left[K \left(\xi \pm 
\frac{\tau}{K}\right)\right]\right\}
\, \, 
\end{equation*}
where $(\tau,\xi) \in {\mathbb{R}}^2$ and 
$\alpha \in {\mathbb{C}}$. 
Therefore, the class of $C^2$-solutions does not provide 
a meaningful framework for discussions of stability, and hence 
some form of functional analytic treatment is necessary, at least
a restriction of the space of admissible data.

\section{A Common Functional Analytic 
Representation of (\ref{eq:1})}

(\ref{eq:1}) results from (\ref{waveequation}) by solution 
for the highest time derivative. In this, the whole equation
is divided by the unbounded function $K^2 e^{-2K\xi}$.
In the next step, (\ref{eq:1}) is represented as a member of 
the class of abstract evolution equations, see 
e.g., \cite{beyer7},
\begin{equation} \label{eq:1new}
u^{\, \prime \prime}(t) = - A \, u(t) \, \, ,
\end{equation}
$t \in {\mathbb{R}}$, where $A : D(A) \rightarrow X$ is 
some densely-defined, linear, positive self-adjoint 
operator in some non-trivial complex Hilbert space 
$\left(X,\braket{|}\right)$. \\

For every equation from this class and for any 
$g_1,g_2$ from the domain $D(A)$ of the corresponding operator 
$A$,
there is a uniquely determined twice continuously 
differentiable map 
$u : {\mathbb{R}} \rightarrow X$ assuming values in $D(A)$ and 
satisfying (\ref{eq:1new})
for all $t \in {\mathbb{R}}$ 
as well as
\begin{equation*}
u(0) = g_1 \, \, , \, \, u^{\, \prime}(0) = g_2 \, \, .
\end{equation*} 
It is important to note that, mainly as a consequence of the 
self-adjointness of $A$, this approach leads automatically 
to a conserved energy.\footnote{This energy corresponds to 
the canonical energy of the classical field $u$, 
described by (\ref{eq:1}), see e.g., \cite{soper}.}  
For this $u$, 
the corresponding canonical energy function 
$E_u : \mathbb{R} \rightarrow {\mathbb{R}}$,
defined by 
\begin{equation*}
E_u(t) := \frac{1}{2} \, \big( \, \braket{u^{\, \prime}(t)| u^{\, \prime}(t)}
+ \braket{u(t)| A u(t)}  \, \big)  
\end{equation*}
for all $t \in {\mathbb{R}}$,
is constant. \\

Finally, if $B:D(B) \rightarrow X$ is some square root of 
$A$, i.e., some densely-defined, linear, self-adjoint 
operator commuting
with $A$ which satisfies 
\begin{equation*}
B^2 = A \, \, ,
\end{equation*}
for example, $B = A^{1/2}$, then this $u$ is given by 
\begin{equation} \label{representation}
u(t) = \cos(tB)g_1 + \frac{\sin(tB)}{B}g_2
\end{equation} 
for all $t \in {\mathbb{R}}$ where $\cos(tB),\sin(tB)/B$ 
denote the bounded
linear operators that are associated by the functional calculus 
for $B$ to the restrictions of 
$\cos$, $\sin/\textrm{id}_{\mathbb{R}}$ to the spectrum of $B$. \\

We note that  
\begin{equation*}
\cos(tB) \, \, , \, \, \frac{\sin(tB)}{B} \, \, ,
\end{equation*}
for every $t \in {\mathbb{R}}$, are bounded linear operators. 
This leads to ``generalized'' solutions of (\ref{eq:1new}) 
for arbitrary data from $X$. For such generalized solutions, 
the corresponding ``energy'' is ill-defined. On the other hand, 
analogously to the Schroedinger equation of quantum theory, where
such generalized solutions are of course physical, 
not only elements from the domain of the Hamilton operator 
are admissible quantum states and are subject to time evolution
\footnote{See e.g., Section~2.1 in \cite{beyer7}.}, 
it does not appear reasonable to discard such 
generalized solutions from consideration. Analogous to 
Schroedinger theory, where the Schroedinger equation is 
merely a ``label'' for the generalized solutions given by 
the corresponding unitary 
one-parameter group, (\ref{eq:1new}) might be considered 
as a ``label'' for (\ref{representation}) and the latter being 
the truly relevant object for applications.  
\\  

Also, we note that (\ref{representation}) implies that  
\begin{equation*}
\|u(t)\| \leq \|g_1\| + |t| \|g_2\|
\end{equation*}
for $t \in {\mathbb{R}}$ and hence {\it that the solutions 
of (\ref{eq:1new}) are stable in the sense that there 
are no exponentially 
growing solutions.} \footnote{Note that the results of this 
section have generalizations to semibounded $A$, 
e.g., see Corollary~2.2.2 in \cite{beyer7}. In particular,
differently to positive $A$, for non-positive $A$, 
there are exponentially growing solutions to (\ref{eq:1new}).} \\

In our special case, (\ref{eq:1}), 
$X = L^2_{\mathbb{C}}({\mathbb{R}})$, 
$A$ is closure of the densely-defined, linear, positive 
symmetric and essentially self-adjoint operator in 
$L^2_{\mathbb{C}}({\mathbb{R}})$
\begin{equation*}
A_0 := 
\begin{pmatrix} 
C_{0}^{\infty}({\mathbb{R}},{\mathbb{C}}) \rightarrow 
 L^2_{\mathbb{C}}({\mathbb{R}})    \\ 
f \mapsto 
- \frac{1}{K^2} f^{\,\prime \prime}
\end{pmatrix} \, \, .
\end{equation*} \\
Application of (\ref{representation}) gives, 
see Theorem~\ref{theorem1} in the Appendix,
a representation of the solutions of (\ref{eq:1new}) given by
\begin{align} \label{representation2}
u(\tau) & = \cos(\tau \bar{p}_{\xi}) g_1 + 
\frac{\sin(\tau \bar{p}_{\xi})}{\bar{p}_{\xi}} \, g_2 \\
& = \frac{1}{2} \left[ 
g_1 \circ \left({\mathrm{id}}_{\mathbb{R}} + \frac{\tau}{K}
\right) + g_1 \circ \left( {\mathrm{id}}_{\mathbb{R}} - \frac{\tau}{K}
\right) + K \sgn(\tau) \left(
\chi_{_{[-|\tau|/K,|\tau|/K]}} * g_2
\right)
\right] \nonumber
\end{align}  
for every $g_1,g_2 \in D(A)$. Here $\circ$ denotes composition, 
${\mathrm{id}}_{\mathbb{R}}$ the identical 
function on ${\mathbb{R}}$, 
\begin{equation*}
\sgn  := 
\chi_{_{(0,\infty)}} - \chi_{_{(-\infty,0)}} \, \, ,
\end{equation*}
$*$ denotes the usual convolution 
product, and $\bar{p}_{\xi}$ is 
the closure of the densely-defined, linear, symmetric 
and essentially self-adjoint operator $p_{\xi}$ in 
$L^2_{\mathbb{C}}({\mathbb{R}})$ given by 
\begin{equation*}
\begin{pmatrix} 
C_{0}^{\infty}({\mathbb{R}},{\mathbb{C}}) \rightarrow 
 L^2_{\mathbb{C}}({\mathbb{R}})    \\ 
f \mapsto 
\frac{i}{K} f^{\,\prime}
\end{pmatrix} \, \, .
\end{equation*} 
As a side remark, 
$\bar{p}_{\xi}$ is a square root of $A$, 
i.e., $\bar{p}_{\xi}^2 = A$, that commutes with 
$A$,
but $\bar{p}_{\xi}$ is different from the positive 
square root, $A^{1/2}$,
of $A$. \\

We note that, as had to be expected, 
essentially (\ref{representation2}) is just a natural 
generalization
of (\ref{sol11wave}) to the elements of the domain $D(A)$ of 
$A$. Since also 
\begin{align*}
& \cos(\tau \bar{p}_{\xi}) g_1 + 
\frac{\sin(\tau \bar{p}_{\xi})}{\bar{p}_{\xi}} \, g_2 \\
& = \frac{1}{2} \left[ 
g_1 \circ \left({\mathrm{id}}_{\mathbb{R}} + \frac{\tau}{K}
\right) + g_1 \circ \left( {\mathrm{id}}_{\mathbb{R}} - \frac{\tau}{K}
\right) + K \sgn(\tau) \left(
\chi_{_{[-|\tau|/K,|\tau|/K]}} * g_2
\right)
\right] \nonumber
\end{align*}  
for every $g_1,g_2 \in L^2_{\mathbb{C}}({\mathbb{R}})$, see 
Theorem~\ref{theorem1} in the Appendix, 
also the generalized solutions 
of (\ref{eq:1new}) in our case are given by a 
natural generalization
of (\ref{sol11wave}) to elements of 
$L^2_{\mathbb{C}}({\mathbb{R}})$.  
\\
 
The slight disadvantage of this common functional analytic 
representation of (\ref{eq:1}) comes from the fact that the data 
for (\ref{eq:1new}) are from $L^2_{\mathbb{C}}({\mathbb{R}})$
and hence vanish in the mean for $\xi \rightarrow -\infty$,
i.e., roughly speaking, {\it vanish in the bifurcation 
point $(0,0)$ of the boundary (horizon) $H$ of $R$},

\begin{equation*}
H = \{(x^0,x^1) \in {\mathbb{R}}^2: x^1 = |x^0|\} \, \, . 
\end{equation*}

Indeed, taking into account an embedding of the Rindler 
wedge into $2$-dimensional Min\-kow\-ski space such behavior 
appears not natural. On the other hand, this behavior 
of the solutions is unsurprising since (\ref{eq:1}) 
is treated analogous to the wave equation on $2$-dimensional 
Minkowski space, where data vanish in the mean at spatial 
infinity. Also, {\it the choice of 
$L^2_{\mathbb{C}}({\mathbb{R}})$ as data space is 
related to the pursued self-adjointness of the operator 
$A$.} On the one hand, the self-adjointness of $A$ enables 
the application
spectral theorems for self-adjoint operators which 
allow the proof of well-posedness of the initial value 
problem for the evolution equation and also lead on a 
representation
of its solutions. Such theorems are generally not available 
for non-self-adjoint operators. Also, along with  
semiboundedness, the self-adjointness of $A$ leads to the 
existence of a conserved energy. \\

This disadvantage of the previous method in connection 
with wave equations on globally-hyperbolic Lorentzian 
space-times with 
horizons has been noticed before, among others by 
\cite{kaywald,beyer1,beyer4}. On the other hand, it needs 
to be stressed that this disadvantage comes into play only 
if it is known that the corresponding space-times are embedded 
in a larger space-time. For instance, the Rindler wedge can be 
embedded into $2$-dimensional Minkowski space. Such embedding 
information is intrinsically ``non-local.'' \\

{\it In addition,
in the physics literature, the solution of wave equations 
for the highest time derivative of the unknown along 
with a subsequent use of the above operator theoretic 
approach is used in most stability discussions, often
implicitly without full realization by the authors, for e.g., 
see \cite{horowitz}.}
Such use is indicated, whenever a stability 
discussion leads
to the finding of unstable eigenvalues/spectra or ``quasinormal 
frequencies.'' For an interpretation of the latter in terms of resonances 
of self-adjoint operators, see \cite{beyer2}.
{\it Apparently, the only rigorous framework  
for such discussion is provided by the spectral theory of 
operators}.
 
\section{A Functional Analytic 
Representation of (\ref{eq:2})}

In this connection, we note that the restriction $\mu|_{R}$ 
of the volume 2-form $\mu = dx^0 \wedge dx^1$ 
on $({\mathbb{R}}^2,g)$ to $R$ is given by 
\begin{equation*}
\mu|_{R} = \frac{e^{2 K \xi}}{K} \, d\tau \wedge d\xi \, \, .
\end{equation*}
Hence 
$\mu|_{R}$ induces on surfaces of constant $\tau$
the $1$-form (``measure'')
\begin{equation*}
\frac{e^{2 K \xi}}{K} \, d\xi \, \, ,
\end{equation*}
where here and in the following $\xi$ is also used as abbreviation 
for ${\textrm{id}}_{\mathbb{R}}$. 
Therefore, we choose 
\begin{equation*}
L^2_{\mathbb{C}}(\mathbb{R},e^{2 K \xi})
\end{equation*} 
as data space for our representation 
of (\ref{eq:2}).
Further, we note that the map $V$ defined by
\begin{equation*}
V f := e^{-K \xi} \cdot f 
\end{equation*} 
for every $f \in L^2_{\mathbb{C}}(\mathbb{R})$ defines 
a Hilbert space isomorphism 
\begin{equation*}
V : L^2_{\mathbb{C}}(\mathbb{R})
\rightarrow L^2_{\mathbb{C}}(\mathbb{R},e^{2 K \xi})
\end{equation*}
with inverse 
\begin{equation*}
V^{-1} = 
\begin{pmatrix} 
L^2_{\mathbb{C}}(\mathbb{R},e^{2 K \xi}) \rightarrow 
 L^2_{\mathbb{C}}({\mathbb{R}})    \\ 
f \mapsto e^{K \xi} f
\end{pmatrix} \, \, .
\end{equation*}
Employing the previous notation, the expression enclosed 
by square brackets in 
(\ref{eq:2}),
\begin{equation*}
e^{-K\xi}  
\left(- \frac{1}{K^2} \,  \partial^2_{\xi} 
\right) e^{K\xi} \, \, ,
\end{equation*}
is represented by 
\begin{equation*}
V A V^{-1} \, \, .
\end{equation*}
The latter operator is densely-defined, linear and positive 
self-adjoint.
The remaining factors $e^{-K\xi}$ in (\ref{eq:2})
are represented by
the corresponding maximal multiplication operator 
$T_{e^{-K\xi}}$ in $L^2_{\mathbb{C}}(\mathbb{R},e^{2 K \xi})$.
We note that $T_{e^{-K\xi}}$ is densely-defined, linear,  
self-adjoint and bijective. Also, $T_{e^{-K\xi}}$ leaves 
$C^{\infty}_{0}({\mathbb{R}},{\mathbb{C}})$ invariant. In this way, 
we arrive at the following functional analytic interpretation 
of (\ref{eq:2}) 
\begin{equation} \label{eq:3}
T_{e^{-K\xi}} \left(T_{e^{-K\xi}} u\right)^{\, \prime \prime} 
= - T_{e^{-K\xi}} VAV^{-1} \left(T_{e^{-K\xi}} u\right) 
\, \, ,
\end{equation}
where primes denote derivatives of paths in 
$L^2_{\mathbb{C}}(\mathbb{R},e^{2 K \xi})$. \\
 
Since $T_{e^{-K\xi}}$ is bijective, the latter equation is 
satisfied if and only if 
\begin{equation*}
\left(T_{e^{-K\xi}} u\right)^{\, \prime \prime} 
= - V A V^{-1} \left(T_{e^{-K\xi}} u\right) \, \, .
\end{equation*}
The latter equation is of type (\ref{eq:1new}).
From the results from the previous section 
as well as the invariance of 
$C^{\infty}_0({\mathbb{R}},{\mathbb{C}})$
under  
\begin{equation*}
\cos(tB) \, \, , \, \, \frac{\sin(tB)}{B} \, \, ,
\end{equation*}
follows for every 
$g_1 \in C^{\infty}_0({\mathbb{R}},{\mathbb{C}})$, 
$g_2 \in C^{\infty}_0({\mathbb{R}},{\mathbb{C}})$, 
that there is a 
unique solution to (\ref{eq:3}) satisfying  
\begin{equation*}
u(0) = g_1 \, \, , \, \, u^{\, \prime}(0) = g_2 
\end{equation*} 
and that this solution is 
given by  
\begin{align} \label{representation3}
u(\tau) & 
= T_{e^{K\xi}} \cos(\tau V \bar{p}_{\xi} V^{-1}) T_{e^{-K\xi}} g_1 + T_{e^{K\xi}} \, 
\frac{\sin(\tau V \bar{p}_{\xi} V^{-1})}{V \bar{p}_{\xi}V^{-1}}
\, T_{e^{-K\xi}} g_2 \\
& = T_{e^{K\xi}} V \cos(\tau \bar{p}_{\xi}) V^{-1} T_{e^{-K\xi}} g_1 + T_{e^{K\xi}}
V \, \frac{\sin(\tau \bar{p}_{\xi})}{\bar{p}_{\xi}} V^{-1}
\, T_{e^{-K\xi}} g_2 \nonumber \\
& = \cos(\tau \bar{p}_{\xi}) g_1 + 
\frac{\sin(\tau \bar{p}_{\xi})}{\bar{p}_{\xi}} \, g_2 \nonumber \\
& = \frac{1}{2} \left[ 
g_1 \circ \left({\mathrm{id}}_{\mathbb{R}} + \frac{\tau}{K}
\right) + g_1 \circ \left( {\mathrm{id}}_{\mathbb{R}} - \frac{\tau}{K}
\right) + K \sgn(\tau) \left(
\chi_{_{[-|\tau|/K,|\tau|/K]}} * g_2 
\right)
\right] \nonumber
\end{align}  
for every $t \in {\mathbb{R}}$, 
where $T_{e^{K\xi}}$ denotes the maximal multiplication operator 
in $L^2_{\mathbb{C}}({\mathbb{R}},e^{2 K \xi})$ corresponding to 
$e^{K\xi}$. Further, for $\tau \in {\mathbb{R}}$, by 
\begin{equation*}
f \circ \left({\mathrm{id}}_{\mathbb{R}} \pm \frac{\tau}{K}
\right) \, \, , \, \, 
\chi_{_{[-|\tau|/K,|\tau|/K]}} * f 
\end{equation*}
for every $f \in L^2_{\mathbb{C}}({\mathbb{R}},e^{2 K \xi})$,
there are defined bounded linear operators on 
$L^2_{\mathbb{C}}({\mathbb{R}},e^{2 K \xi})$, see 
Theorems~\ref{theorem2},~\ref{theorem3} in the Appendix. \\

As had to be expected, also (\ref{eq:3}) turns out as 
``label'' for the natural generalization 
of (\ref{sol11wave}) to the 
elements $L^2_{\mathbb{C}}({\mathbb{R}},e^{2 K \xi})$. \\

Still the functional analytic representations of (\ref{eq:1})
and (\ref{eq:2}) are different, since the field $u$, as a 
classical field, is observable. Also the field corresponding 
to (\ref{eq:2}) does not necessarily vanish in the bifurcation 
point $(0,0)$ of the horizon. For instance, the restriction 
of the generalized solution of (\ref{eq:2}) to 
\begin{equation*}
W_{L} := \left\{(\tau,\xi) \in {\mathbb{R}}^2: \xi < 
\frac{|\tau|}{K}\right\} 
\end{equation*}
corresponding to data 
\begin{equation*}
u(0,\cdot) = e^{-\alpha \xi} \cdot \chi_{_{(-\infty,0)}}
\, \, , \, \, u^{\prime}(0,\cdot) = 0 \, \, ,
\end{equation*}
where $\alpha < K$,
is given by 
\begin{equation} \label{specialsolution}
u(\tau,\xi) = e^{-\alpha \xi} \cosh\left(\frac{\alpha \tau}{K}
\right)
\end{equation}
for every $(\tau,\xi) \in W_{L}$. The latter leads to 
\begin{equation*}
u(x^0,x^1) = \frac{1}{2} \left\{ 
[K(x^1+x^0)]^{-\frac{\alpha}{K}} + 
[K(x^1-x^0)]^{-\frac{\alpha}{K}} 
\right\}
\end{equation*} 
for every $(x^0,x^1) \in {\mathbb{R}}^2$ satisfying 
\begin{equation*}
|x^0| < x^1 < |x^0| + \frac{1}{K} \, \, .
\end{equation*}
In addition, {\it the 
weighted $L^2$-norm that is corresponding to 
(\ref{specialsolution}) is exponentially increasing in $\tau$}.
Therefore, the solutions to (\ref{eq:3}) might be considered 
unstable, {\it on the other hand with respect to inertial 
coordinate system $(x^0,x^1)$, 
no exponential growth is visible}. Of course, such outcome 
also raises the question of coordinate dependence of the 
results.

\section{Discussion}
\label{discussion}

This brief points out that the
formulation of physically reasonable initial-boundary
value problems for wave equations in Lorentzian space-times
is not unique, i.e., that there are inequivalent such 
formulations that lead to a different outcome 
of the stability discussion of the solutions. \\
 
For the example of the wave equation on the right 
Rindler wedge in $2$-dimensional Minkowski space, this 
note gives $2$ inequivalent 
formulations of a well-posed initial-boundary 
value problem, leading to different outcomes of
the stability discussion of the solutions. Their  
construction suggest the existence of many more 
of such formulations. All what needs to be found 
is a weight $\rho$ such that the family of linear 
operators on $C^{\infty}_{0}({\mathbb{R}},{\mathbb{C}})$, 
defined by  
\begin{equation*}
\frac{1}{2} \left[ 
\varphi \circ \left({\mathrm{id}}_{\mathbb{R}} + \frac{\tau}{K}
\right) + \varphi \circ \left( {\mathrm{id}}_{\mathbb{R}} - \frac{\tau}{K}
\right)\right] \, \, , \, \, 
\frac{K}{2} \sgn(\tau) \left(
\chi_{_{[-|\tau|/K,|\tau|/K]}} * \varphi \right)
\end{equation*}
for every $\varphi \in 
C^{\infty}_{0}({\mathbb{R}},{\mathbb{C}})$, 
where $\tau$ runs through the elements of ${\mathbb{R}}$,
from the classical solution formula (\ref{sol11wave}),
lead on bounded linear operators in 
$L^2_{\mathbb{C}}({\mathbb{R}},\rho)$. For such $\rho$, 
according to the linear 
extension theorem, every member of the family has 
a unique extension to a bounded linear operator 
on $L^2({\mathbb{R}},\rho)$, and the resulting family 
of bounded linear operators on $L^2({\mathbb{R}},\rho)$
can be viewed as resulting from a functional analytic 
interpretation 
of (\ref{waveequation}). \\

The used methods can be generalized to
wave equations on stationary globally hyperbolic space-times 
with horizons in higher dimensions, such 
Schwarzschild and Kerr space-times.

\section*{Acknowledgments}
 
H.B. is thankful for the hospitality and support
of 
the `Division for Theoretical Astrophysics' (TAT,  
K. Kokkotas) of the Institute for Astronomy and 
Astrophysics   
at the Eber\-hard-Karls-University Tuebingen. 
This work was supported by SNI-M\'exico, and the 
SFB/Trans\-re\-gi\-o~$7$ on 
``Gravitational Wave Astronomy'' of the German Science 
Foundation (DFG).
 
\section{Appendix}

\begin{ass}
In the following, we denote for every $n \in {\mathbb{N}}^{*}$ 
by $v^n$ the Lebesgue measure on 
${\mathbb{R}}^n$,  
${\cal S}_{\mathbb{C}}({\mathbb{R}})$ the space 
of rapidly decreasing test functions on ${\mathbb{R}}$,
${\mathrm{id}}_{\mathbb{R}}$ the identical 
function on ${\mathbb{R}}$, 
$F_0 : {\cal S}_{\mathbb{C}}({\mathbb{R}}) \rightarrow 
{\cal S}_{\mathbb{C}}({\mathbb{R}})$ the Fourier transformation, 
defined by  
\begin{equation*}
F_0(f)(k) := (2 \pi)^{-1/2} \int_{\mathbb{R}} e^{- i k . 
{\textrm{id}}_{\mathbb{R}}} f dv^1 
\end{equation*}
for every $k \in {\mathbb{R}}$ and $f 
\in {\cal S}_{\mathbb{C}}({\mathbb{R}})$, 
$F_1 : L^1_{\mathbb{C}}({\mathbb{R}}) \rightarrow
C_{\infty}({\mathbb{R}},{\mathbb{C}})$ 
the Fourier transformation defined by 
\begin{equation*}
F_1(f)(k) := (2 \pi)^{-1/2} \int_{\mathbb{R}} e^{- i k . 
{\textrm{id}}_{\mathbb{R}}} f dv^1 
\end{equation*}
for every $k \in {\mathbb{R}}$ and $f \in 
L^1_{\mathbb{C}}({\mathbb{R}})$,
and by 
$F_2$ the unitary Fourier 
transformation that is induced by $F_0$ on 
$L^2_{\mathbb{C}}({\mathbb{R}})$. Further, for every complex-valued
function that is a.e. defined and measurable on 
${\mathbb{R}}$, we denote by $T_{g}$ the maximal multiplication
operator with $g$ in $L^2_{\mathbb{C}}({\mathbb{R}})$.
Finally, $A$ is closure of the densely-defined, linear, positive 
symmetric and essentially self-adjoint operator in 
$L^2_{\mathbb{C}}({\mathbb{R}})$
\begin{equation*}
A_0 := 
\begin{pmatrix} 
C_{0}^{\infty}({\mathbb{R}},{\mathbb{C}}) \rightarrow 
 L^2_{\mathbb{C}}({\mathbb{R}})    \\ 
f \mapsto 
- \frac{1}{K^2} \, f^{\,\prime \prime}
\end{pmatrix} \, \, , 
\end{equation*}
and $\bar{p}_{\xi}$ is 
the closure of the densely-defined, linear, symmetric 
and essentially self-adjoint operator $p_{\xi}$ in 
$L^2_{\mathbb{C}}({\mathbb{R}})$ given by 
\begin{equation*}
p_{\xi} := 
\begin{pmatrix} 
C_{0}^{\infty}({\mathbb{R}},{\mathbb{C}}) \rightarrow 
 L^2_{\mathbb{C}}({\mathbb{R}})    \\ 
f \mapsto 
\frac{i}{K} f^{\,\prime}
\end{pmatrix} \, \, .
\end{equation*}
In particular, the spectrum of $\bar{p}_{\xi}$ is given by 
${\mathbb{R}}$, and $\bar{p}_{\xi}$ is the infinitesimal 
generator of the strongly continuous one-parameter 
unitary group $U : {\mathbb{R}} \rightarrow 
L(L^2_{\mathbb{C}}({\mathbb{R}}),L^2_{\mathbb{C}}({\mathbb{R}}))$
given by 
\begin{equation*}
U(\tau) := \begin{pmatrix} 
L^2_{\mathbb{C}}({\mathbb{R}}) \rightarrow 
 L^2_{\mathbb{C}}({\mathbb{R}})    \\ 
f \mapsto 
f \circ \left({\textrm{id}}_{\mathbb{R}} - \frac{\tau}{K} \right)
\end{pmatrix}  \left( = e^{ i \tau \bar{p}_{\xi}}  \right) 
\end{equation*}
for every $\tau \in {\mathbb{R}}$.
\end{ass}

\begin{theorem} \label{theorem1}
Then  
\begin{equation*}
\bar{p}_{\xi}^2 = A \, \, , 
\end{equation*}
$\bar{p}_{\xi}$ commutes with $A$ in the strong sense,
and  
\begin{align} \label{representation2a}
\cos(\tau \bar{p}_{\xi}) g_1 & = \frac{1}{2} \left[ 
g_1 \circ \left({\mathrm{id}}_{\mathbb{R}} + \frac{\tau}{K}
\right) + g_1 \circ \left( {\mathrm{id}}_{\mathbb{R}} - \frac{\tau}{K}
\right)\right] \, \, , \\
\frac{\sin(\tau \bar{p}_{\xi})}{\bar{p}_{\xi}} \, g_2 
& = \frac{K}{2} \, \sgn(\tau) \left(
\chi_{_{[-|\tau|/K,|\tau|/K]}} * g_2
\right) \nonumber
\end{align}  
for every $g_1,g_2 \in L^2_{\mathbb{C}}({\mathbb{R}})$, where 
$\circ$ denotes composition,  
\begin{equation*}
\sgn  := 
\chi_{_{(0,\infty)}} - \chi_{_{(-\infty,0)}} \, \, ,
\end{equation*}
and $*$ denotes the usual convolution 
product.
\end{theorem}

\begin{proof}
As a square of a densely-defined, linear and self-adjoint 
operator, $\bar{p}_{\xi}^2$ is a densely-defined, linear 
and self-adjoint operator in $L^2_{\mathbb{C}}({\mathbb{R}})$.
Obviously, it follows that $\bar{p}_{\xi}^2 \supset A_0$ 
and hence, since $A_0$ is in particular essentially self-adjoint,
that $\bar{p}_{\xi}^2 = \bar{A}_0 = A$. \\

Further, it follows that
\begin{equation*}
\exp(i \tau \bar{p}_{\xi}) A_0 \subset A_0 
\exp(i \tau \bar{p}_{\xi})  \, \, . 
\end{equation*}
For the proof, we note that 
\begin{align*}
& \exp(i \tau \bar{p}_{\xi}) A_0 f = 
- \frac{1}{K^2}  \, \exp(i \tau \bar{p}_{\xi}) 
f^{\,\prime \prime} = - \frac{1}{K^2}  \left( 
f^{\,\prime \prime} \circ \left({\textrm{id}}_{\mathbb{R}} - \frac{\tau}{K} \right)
\right) \\
& = - \frac{1}{K^2}  \left( 
f \circ \left({\textrm{id}}_{\mathbb{R}} - \frac{\tau}{K} \right)
\right)^{\,\prime \prime} = A_0 \left( 
f \circ \left({\textrm{id}}_{\mathbb{R}} - \frac{\tau}{K} \right)
\right) = A_0 \exp(i \tau \bar{p}_{\xi}) f \, \, .
\end{align*}
for $\tau \in {\mathbb{R}}$ and $f \in 
C_{0}^{\infty}({\mathbb{R}},{\mathbb{C}})$.
Hence it follows also that 
\begin{equation*}
\exp(i \tau \bar{p}_{\xi}) A \subset A
\exp(i \tau \bar{p}_{\xi})   
\end{equation*}
which implies that 
$\bar{p}_{\xi}$ commutes with $A$ in the strong sense. \\

For the proof of (\ref{representation2a}), we note that
for $\tau \in [0,\infty)$, $g \in L^2_{\mathbb{C}}({\mathbb{R}})$ 
and $h \in C_{0}^{\infty}({\mathbb{R}},{\mathbb{C}})$ 
\begin{equation} \label{convolution1}
\chi_{_{[-\tau,\tau]}} * g 
\in C_{\infty}({\mathbb{R}},{\mathbb{C}}) \cap 
L^2_{\mathbb{C}}({\mathbb{R}}) 
\end{equation} 
and that 
\begin{equation} \label{convolution2}
\left(\chi_{_{[-\tau,\tau]}} \circ p_2\right) \cdot 
\left(g \circ (p_1 - p_2)\right) \cdot \left(h \circ p_1\right)
\end{equation} 
is $v^2$-summable, where $p_1,p_2$ denote the coordinate 
projections of ${\mathbb{R}}^2$ onto the first and 
second coordinate, 
respectively. For the proof, let 
$\tau \in [0,\infty)$, $g \in L^2_{\mathbb{C}}({\mathbb{R}})$ and 
$h \in C_{0}^{\infty}({\mathbb{R}},{\mathbb{C}})$. Since 
$\chi_{_{[-\tau,\tau]}} \in L^2_{\mathbb{C}}({\mathbb{R}})$ and 
\begin{equation*}
F_2 \, \chi_{_{[-\tau,\tau]}} = \sqrt{\frac{2}{\pi}} \, 
\widehat{\frac{\sin(\tau . {\textrm{id}}_{\mathbb{R}})}{{\textrm{id}}_{\mathbb{R}}}} \, \, , 
\end{equation*}  
where 
\begin{equation*}
\widehat{\frac{\sin(\tau . {\textrm{id}}_{\mathbb{R}})}{{\textrm{id}}_{\mathbb{R}}}}
\end{equation*}
denotes the extension of 
\begin{equation*}
\frac{\sin(\tau . {\textrm{id}}_{\mathbb{R}})}{{\textrm{id}}_{\mathbb{R}}} \in C({\mathbb{R}}^{*},{\mathbb{C}})
\end{equation*}
to an element of $C_{\infty}({\mathbb{R}},{\mathbb{C}})$, it follows
that 
\begin{equation*}
\chi_{_{[-\tau,\tau]}} * g = F_1\left( \sqrt{\frac{2}{\pi}} \, 
\widehat{\frac{\sin(\tau . {\textrm{id}}_{\mathbb{R}})}{{\textrm{id}}_{\mathbb{R}}}} \cdot \left[ 
(F_2g) \circ (- {\textrm{id}}_{\mathbb{R}}) 
\right]\right) \, \, .
\end{equation*}
Since in particular, the argument of $F_1$ in the 
latter equality is also contained 
in $L^2_{\mathbb{C}}({\mathbb{R}})$, (\ref{convolution1})
follows. Further, if $N \in {\mathbb{N}}^{*}$ is such that 
$\textrm{supp}(h) \subset [-N,N]$, then 
\begin{align*}
& \left(\chi_{_{[-\tau,\tau]}} \circ p_2\right) \cdot 
(g \circ (p_1-p_2)) \cdot (h \circ p_1) \\
& =
\left(\chi_{_{[-\tau,\tau]}} \circ p_2\right) \cdot
\left[
\left(\chi_{_{[-(N + \tau),(N + \tau)]}} \cdot g \right)
\circ (p_1 - p_2)
\right] \cdot (h \circ p_1) \, \, .
\end{align*}
Since $\chi_{_{[-(N + \tau),(N + \tau)]}} \in 
L^2_{\mathbb{C}}({\mathbb{R}})$, it follows that 
$\chi_{_{[-(N + \tau),(N + \tau)]}} \cdot g \in 
L^1_{\mathbb{C}}({\mathbb{R}})$ and hence, since also
$\chi_{_{[-\tau,\tau]}} \in 
L^1_{\mathbb{C}}({\mathbb{R}})$, it follows from 
a known result in connection with convolution products
that 
\begin{equation*}
\left(\chi_{_{[-\tau,\tau]}} \circ p_2\right) \cdot
\left[
\left(\chi_{_{[-(N + \tau),(N + \tau)]}} \cdot g \right)
\circ (p_1 - p_2)
\right]
\end{equation*} 
is $v^2$-summable. Finally, since $h \circ p_1$ is
bounded and continuous, it follows from integration theory 
that (\ref{convolution2}) is $v^2$-summable. \\

Finally, we obtain from direct calculation that 
\begin{equation*}
\cos(\tau \bar{p}_{\xi}) f = \frac{1}{2} 
\left[
\exp(i \tau \bar{p}_{\xi}) + \exp(-i \tau \bar{p}_{\xi})
\right]f  = \frac{1}{2} 
\left[f \circ \left({\textrm{id}}_{\mathbb{R}} - \frac{\tau}{K} \right) + f \circ \left({\textrm{id}}_{\mathbb{R}} + \frac{\tau}{K} \right) 
\right] 
\end{equation*}
as well as 
\begin{align*}
& \braket{\frac{\sin(\tau \bar{p}_{\xi})}{\bar{p}_{\xi}}\,g|h} =
\textrm{sgn}(\tau)  
\braket{\frac{\sin(|\tau| \bar{p}_{\xi})}{\bar{p}_{\xi}}\,g|h}
= \textrm{sgn}(\tau)
\braket{g|F_1\!\left(\frac{1}{2}\, 
\chi_{_{[-|\tau|,|\tau|]}}\right)\!(\bar{p}_{\xi})h} \\
& = \frac{1}{2} \, \textrm{sgn}(\tau) 
\int_{-|\tau|}^{|\tau|} \braket{\exp(it\bar{p}_{\xi})g|h} dt
= \frac{K}{2} \, \textrm{sgn}(\tau) 
\int_{-|\tau|/K}^{|\tau|/K} \braket{g \circ ({\textrm{id}}_{\mathbb{R}} - t)|h} dt \\
& = \frac{K}{2} \, \textrm{sgn}(\tau)
\int_{{\mathbb{R}}^2}
\left(\chi_{_{[-|\tau|/K,|\tau|/K]}} \circ p_2\right) \cdot 
(g^{*} \circ (p_1-p_2)) \cdot (h \circ p_1) dv^2 \\
& = \frac{K}{2} \, \textrm{sgn}(\tau) 
\int_{\mathbb{R}} \left(\chi_{_{[-|\tau|/K,|\tau|/K]}} * g 
\right)^{*} \cdot h dv^1 = 
\braket{\frac{K}{2} \, \textrm{sgn}(\tau) \left(\chi_{_{[-|\tau|/K,|\tau|/K]}} * g 
\right)|h} \, \, , 
\end{align*}
for $\tau \in {\mathbb{R}}$, $f, g \in L^2_{\mathbb{C}}({\mathbb{R}})$ and $h \in C_{0}^{\infty}({\mathbb{R}},{\mathbb{C}})$. Since
$C_{0}^{\infty}({\mathbb{R}},{\mathbb{C}})$ is dense in 
$L^2_{\mathbb{C}}({\mathbb{R}})$, from the latter follows that 
\begin{equation*}
\frac{\sin(\tau \bar{p}_{\xi})}{\bar{p}_{\xi}}\,g = 
\frac{K}{2} \, \textrm{sgn}(\tau) 
\chi_{_{[-|\tau|/K,|\tau|/K]}} * g \, \, .
\end{equation*}
\end{proof}

\begin{ass}
In addition, we denote by $B^{\pm 1/2}$ the 
multiplikation operator by $T_{e^{\mp K\xi}}$ in 
$L^2_{\mathbb{C}}({\mathbb{R}})$. Further, 
we denote by $U_{\mathbb{C}}^s({\mathbb{R}})$
the set of bounded complex-valued functions on ${\mathbb{R}}$
with component functions that are strongly measurable, in 
the sense that they are everywhere
on ${\mathbb{R}}$ the limit of a sequence of 
step functions. Finally, ${\cal A}$ denotes the set of 
all complex-valued functions $f$ 
on ${\mathbb{R}} \times [-1,0]$
satisfying
\begin{enumerate}
\item $f(\cdot,0) , f(\cdot,1) \in U_{\mathbb{C}}^s({\mathbb{R}})$, 
\item a.e. on ${\mathbb{R}}$:
\begin{equation*}
\lim_{y \rightarrow 0-} f(\cdot,y) =  f(\cdot,0) \, \, , \, \, 
\lim_{y \rightarrow (-1)+} f(\cdot,y) =  f(\cdot,-1) \, \, ,
\end{equation*} 
\item $f|_{{\mathbb{R}} \times (-1,0)}$ is 
holomorphic, and there 
are $C \geq 0$, $N \in {\mathbb{N}}$ such that 
\begin{equation*}
|f(z)| \leq C (1+|z|)^N
\end{equation*}
for every $z \in {\mathbb{R}} \times [-1,0]$.
\end{enumerate} 
\end{ass}

\begin{theorem} \label{theorem2}
For every $f \in {\cal A}$
\begin{equation*}
C^{\infty}_{0}({\mathbb{R}},{\mathbb{C}}) 
\subset D(B^{-1/2} f(\bar{p}_{\xi},0)B^{1/2}) \, \, 
\textrm{and} \, \, \overline{B^{-1/2} f(\bar{p}_{\xi},0)B^{1/2}} 
= f(\bar{p}_{\xi},-1) \, \, ,
\end{equation*}
where the overline on top of the expression containing 
$f(\bar{p}_{\xi},0)$ indicates closure in the operator norm of
$L(L^2_{\mathbb{C}}({\mathbb{R}}),L^2_{\mathbb{C}}({\mathbb{R}}))$.
\end{theorem}

\begin{proof}
In a first step, we prove an auxiliary result. For this purpose,
let $f \in {\cal A}$, $\varphi \in C({\mathbb{R}} \times [-1,0],{\mathbb{C}})$ such that 
$\varphi|_{{\mathbb{R}} \times (-1,0)}$ is holomorphic and 
$({\textrm{id}}_{\mathbb{C}})^{k} \cdot \varphi$ is bounded 
for every $k \in {\mathbb{N}}$. Then, 
\begin{equation*}
\int_{\mathbb{R}} f(\cdot,0) \cdot \varphi(\cdot,0) \, dv^1 =
\int_{\mathbb{R}} f(\cdot,1) \cdot \varphi(\cdot,1) \, dv^1
\, \, .
\end{equation*} 
The proof is a straighforward application of Cauchy's 
integral theorem for rectangular paths and Lebesgue's 
dominated convergence theorem. \\

Further, we note that from a well-known theorem of Paley-Wiener, 
see e.g., \cite{reedsimon} Vol. I, Sect IX.3, 
follows that for any $g \in 
C^{\infty}_{0}({\mathbb{R}},{\mathbb{C}})$ the corresponding
function 
\begin{equation*}
\begin{pmatrix} 
{\mathbb{R}} \times [-1,0] \rightarrow {\mathbb{C}} \\ 
z \mapsto \int_{\mathbb{R}} e^{-i z . {\textrm{id}}_{\mathbb{R}}}
\cdot g \, dv^1 
\end{pmatrix} 
\end{equation*}
has the properties that were required for $\varphi$ in the 
previous auxiliary result. \\

As a further auxiliary result, we note that 
\begin{equation*}
\begin{pmatrix} 
{\mathbb{R}} \rightarrow {\mathbb{C}} \\ 
\tau \mapsto \braket{h_1|\exp(i \tau \bar{p}_{\xi}) h_2}
\end{pmatrix}
\in C^{\infty}_{0}({\mathbb{R}},{\mathbb{C}})
\end{equation*} 
for every $h_1, h_2 \in C^{\infty}_{0}({\mathbb{R}},{\mathbb{C}})$.
The latter follows from the identities 
\begin{align*}
& \braket{h_1|\exp(i \tau \bar{p}_{\xi}) h_2} =
\braket{h_1|h_2 \circ \left({\textrm{id}}_{\mathbb{R}} - 
\frac{\tau}{K}\right)} = 
\int_{\mathbb{R}} h^{*}_1 \left[h_2 \circ \left({\textrm{id}}_{\mathbb{R}} - 
\frac{\tau}{K}\right) \right] dv^1 \\
& = \frac{1}{K} \, 
\int_{\mathbb{R}} \left(h^{*}_1 \circ \left(K^{-1} . 
{\textrm{id}}_{\mathbb{R}}\right)\right) 
\left[h_2 \circ \left(- K^{-1} . 
{\textrm{id}}_{\mathbb{R}}\right) \circ \left(\tau - 
{\textrm{id}}_{\mathbb{R}} \right) \right] dv^1 \\
& = \frac{1}{K} \left[ 
\left(h^{*}_1 \circ \left(K^{-1} . 
{\textrm{id}}_{\mathbb{R}}\right)\right) * 
\left(h_2 \circ \left(- K^{-1} . 
{\textrm{id}}_{\mathbb{R}}\right)\right) 
\right](\tau) \\
& = \frac{\sqrt{2 \pi}}{K} \, F_0 \left[ 
F_0^{-1} \left(h^{*}_1 \circ \left(K^{-1} . 
{\textrm{id}}_{\mathbb{R}}\right)\right) \cdot  
F_0^{-1} \left(h_2 \circ \left(-K^{-1} . 
{\textrm{id}}_{\mathbb{R}}\right)\right) 
\right](\tau)
\end{align*}
which show that 
\begin{equation*}
\begin{pmatrix} 
{\mathbb{R}} \rightarrow {\mathbb{C}} \\ 
\tau \mapsto \braket{h_1|\exp(i \tau \bar{p}_{\xi}) h_2}
\end{pmatrix}
\in C_{0}({\mathbb{R}},{\mathbb{C}}) \cap 
{\cal S}_{\mathbb{C}}({\mathbb{R}}) 
= C_{0}^{\infty}({\mathbb{R}},{\mathbb{C}}) \, \, .
\end{equation*}
As a final auxilary result, we note that 
$C_{0}^{\infty}({\mathbb{R}},{\mathbb{C}})$ is a core 
for $B^{-1/2}$. The latter follows from the facts 
that 
$C_{0}^{\infty}({\mathbb{R}},{\mathbb{C}})$ is dense 
in $L^2_{\mathbb{C}}({\mathbb{R}})$, contained 
in $D(B^{-1/2})$ and invariant under the 
strongly continuous one-parameter 
unitary group 
\begin{equation*}
\begin{pmatrix} 
{\mathbb{R}} \rightarrow 
L(L^2_{\mathbb{C}}({\mathbb{R}}),L^2_{\mathbb{C}}({\mathbb{R}})) \\ 
t \mapsto T_{\exp(it.\exp{(K.{\textrm{id}}_{\mathbb{R}})})} 
\end{pmatrix}
\end{equation*}
that is generated by $B^{-1/2}$, where for every 
$t \in {\mathbb{R}}$, the corresponding 
$T_{\exp(it.\exp{(K.{\textrm{id}}_{\mathbb{R}})})}$ denotes the maximal multiplication operator 
in $L^2_{\mathbb{C}}({\mathbb{R}})$ by the function 
$\exp(it.\exp{(K.{\textrm{id}}_{\mathbb{R}})})$. 
\\

For the proof of our main result, let $f \in {\cal A}$. Since
$B^{-1/2}$ is in particular self-adjoint and 
$C_{0}^{\infty}({\mathbb{R}},{\mathbb{C}})$ is a core 
for $B^{-1/2}$, it follows for $h_1 \in 
C_{0}^{\infty}({\mathbb{R}},{\mathbb{C}})$
that 
\begin{equation*}
f({\bar{p}}_{\xi},0) B^{1/2} h_1 \in D(B^{-1/2})
\end{equation*}
if and only if the linear form  
\begin{equation*}
\braket{f({\bar{p}}_{\xi},0) B^{1/2} h_1|B^{-1/2}\cdot}|_{C_{0}^{\infty}({\mathbb{R}},{\mathbb{C}})}
\end{equation*}
is bounded. With the help of the previous auxiliar results, we 
conclude for 
$h_1,h_2 \in C_{0}^{\infty}({\mathbb{R}},{\mathbb{C}})$ that 
\begin{align*}
& \left(\braket{f({\bar{p}}_{\xi},0)B^{1/2} 
h_1|B^{-1/2} h_2} \right)^{*} \\
& = 
\frac{1}{2\pi} \, 
\int_{\mathbb{R}} f(\cdot,0) F_1 
\begin{pmatrix} 
{\mathbb{R}} \rightarrow {\mathbb{C}} \\ 
\tau \mapsto \braket{B^{-1/2} h_2|\exp(i \tau \bar{p}_{\xi}) 
B^{1/2}h_1}
\end{pmatrix}
dv^1 \\
& = \frac{1}{2\pi} \, 
\int_{\mathbb{R}} f(\cdot,0) F_1 e^{{\mathrm{id}}_{\mathbb{R}}}
\begin{pmatrix} 
{\mathbb{R}} \rightarrow {\mathbb{C}} \\ 
\tau \mapsto \braket{h_2|\exp(i \tau \bar{p}_{\xi}) h_1}
\end{pmatrix}
dv^1 \\
& = \frac{1}{2\pi} \,
\int_{\mathbb{R}} f(\cdot,-1) F_1
\begin{pmatrix} 
{\mathbb{R}} \rightarrow {\mathbb{C}} \\ 
\tau \mapsto \braket{h_2|\exp(i \tau \bar{p}_{\xi}) h_1}
\end{pmatrix}
dv^1 \\
& = \braket{h_2|f({\bar{p}}_{\xi},-1)h_1} =
\left(\braket{f({\bar{p}}_{\xi},-1)h_1|h_2}\right)^{*}
\, \, . 
\end{align*}
From the latter, we conclude from the self-adjointness of 
$B^{-1/2}$ that for 
$h_1 \in C_{0}^{\infty}({\mathbb{R}},{\mathbb{C}})$ that
\begin{equation*} 
f({\bar{p}}_{\xi},0) B^{1/2} h_1 \in D(B^{-1/2})
\end{equation*}
and that 
\begin{equation*} 
B^{-1/2}f({\bar{p}}_{\xi},0) B^{1/2} h_1 = 
f({\bar{p}}_{\xi},-1) h_1 \, \, .
\end{equation*} 
Finally, 
since $C_{0}^{\infty}({\mathbb{R}},{\mathbb{C}})$ is dense in 
$L^2_{\mathbb{C}}({\mathbb{R}})$, the latter shows that 
$B^{-1/2}f({\bar{p}}_{\xi},0) B^{1/2}$ is a densely-defined, linear 
and bounded operator in $L^2_{\mathbb{C}}({\mathbb{R}})$,
whose extension to a bounded linear operator on
$L^2_{\mathbb{C}}({\mathbb{R}})$ is given by 
$f({\bar{p}}_{\xi},-1)$.
\end{proof}

\begin{theorem} \label{theorem3}
\begin{align} \label{representation3a}
V \circ  \overline{B^{-1/2}\cos(\tau \bar{p}_{\xi})B^{1/2}} \, 
V^{-1} f
& = \frac{1}{2} \left[ 
f \circ \left({\mathrm{id}}_{\mathbb{R}} + \frac{\tau}{K}
\right) + f \circ \left( {\mathrm{id}}_{\mathbb{R}} - \frac{\tau}{K}
\right)\right] \, \, , \nonumber 
\\
V \circ  \overline{B^{-1/2}\, 
\frac{\sin(\tau \bar{p}_{\xi})}{\bar{p}_{\xi}} B^{1/2}} \, V^{-1} g 
& = \frac{K}{2} \, \sgn(\tau) \left(
\chi_{_{[-|\tau|/K,|\tau|/K]}} * g
\right)
\end{align}
for all $f, g \in L^2_{\mathbb{C}}({\mathbb{R}},e^{2 K \xi})$, 
where the overline on top of the expressions starting with 
$B^{-1/2}$
indicates closure in the operator norm of
$L(L^2_{\mathbb{C}}({\mathbb{R}}),L^2_{\mathbb{C}}({\mathbb{R}}))$.
\end{theorem}

\begin{proof}
In a first step, we note that 
\begin{equation*}
\cos(\tau . 
{\textrm{id}}_{\mathbb{R}}) 
\, \, , \, \, 
\widehat{\frac{\sin(\tau . {\textrm{id}}_{\mathbb{R}})}{{\textrm{id}}_{\mathbb{R}}}} \in {\cal A} \, \, .
\end{equation*}
Hence according to Theorem~\ref{theorem2},
the linear operators  
\begin{equation*}
B^{-1/2}\cos(\tau \bar{p}_{\xi})B^{1/2} \, \, , \, \, 
B^{-1/2}\, 
\frac{\sin(\tau \bar{p}_{\xi})}{\bar{p}_{\xi}} B^{1/2}
\end{equation*}
are in particular densely-defined, with domains containing 
$C_{0}^{\infty}({\mathbb{R}},{\mathbb{C}})$, and bounded. 
As a consequence, 
these operators have unique extensions to bounded linear operators
on $L^2_{\mathbb{C}}({\mathbb{R}})$. \\

Further, 
\begin{equation*}
B^{1/2} V^{-1} \varphi = e^{-K \xi} e^{K \xi } \varphi = 
\varphi \, \, , \, \, 
V B^{-1/2} f = e^{-K \xi} e^{K \xi } f = f 
\end{equation*}
for every 
$\varphi \in C_{0}^{\infty}({\mathbb{R}},{\mathbb{C}})$
and 
$f \in D(B^{-1/2})$. Hence it follows from Theorem~\ref{theorem1}
that 
\begin{align*} 
V \circ  B^{-1/2} \cos(\tau \bar{p}_{\xi}) B^{1/2}  
V^{-1} \varphi
& = \frac{1}{2} \left[ 
\varphi \circ \left({\mathrm{id}}_{\mathbb{R}} + \frac{\tau}{K}
\right) + \varphi \circ \left( {\mathrm{id}}_{\mathbb{R}} - \frac{\tau}{K}
\right)\right] \, \, , \nonumber 
\\
V \circ  B^{-1/2} \, 
\frac{\sin(\tau \bar{p}_{\xi})}{\bar{p}_{\xi}} B^{1/2} \, V^{-1} \varphi 
& = \frac{K}{2} \, \sgn(\tau) \left(
\chi_{_{[-|\tau|/K,|\tau|/K]}} * \varphi
\right)
\end{align*}
for every   
$\varphi \in C_{0}^{\infty}({\mathbb{R}},{\mathbb{C}})$. \\

Finally, for $\tau \in {\mathbb{R}}$, by 
\begin{equation*}
f \circ \left({\mathrm{id}}_{\mathbb{R}} \pm \frac{\tau}{K}
\right) \, \, , \, \, 
\chi_{_{[-|\tau|/K,|\tau|/K]}} * f 
\end{equation*}
for every $f \in L^2_{\mathbb{C}}({\mathbb{R}},e^{2 K \xi})$,
there are defined bounded linear operators on 
$L^2_{\mathbb{C}}({\mathbb{R}},e^{2 K \xi})$. \\

For the proof, we note that 
\begin{equation*}
e^{\mp \tau} . V \circ e^{i(\mp \tau) {\bar{p}}_{\xi}} \circ V^{-1} f = f \circ \left({\mathrm{id}}_{\mathbb{R}} \pm \frac{\tau}{K}
\right)
\end{equation*}
for every $f \in L^2_{\mathbb{C}}({\mathbb{R}},e^{2 K \xi})$.
Also, by 
\begin{equation*}
\begin{pmatrix} 
L^2_{\mathbb{C}}({\mathbb{R}}) \rightarrow 
L^2_{\mathbb{C}}({\mathbb{R}}) \\ 
f \mapsto (\exp(K . {\textrm{id}}_{\mathbb{R}}) \cdot 
\chi_{_{[-|\tau|/K,|\tau|/K]}}) * f 
\end{pmatrix}
\end{equation*}
there is defined a bounded linear operator on 
$L^2_{\mathbb{C}}({\mathbb{R}})$. The latter can be seen as follows.
Since $\exp(K . {\textrm{id}}_{\mathbb{R}}) \cdot 
\chi_{_{[-|\tau|/K,|\tau|/K]}} \in 
L^2_{\mathbb{C}}({\mathbb{R}})$, 
\begin{equation*}
\left(\exp(K . {\textrm{id}}_{\mathbb{R}}) \cdot 
\chi_{_{[-|\tau|/K,|\tau|/K]}})\right) * f \in C_{\infty}({\mathbb{R}}, {\mathbb{C}}) \cap 
L^2_{\mathbb{C}}({\mathbb{R}}) 
\end{equation*}
Further, 
\begin{equation*}
\|\left(\exp(K . {\textrm{id}}_{\mathbb{R}}) \cdot 
\chi_{_{[-|\tau|/K,|\tau|/K]}})\right) * f \|_2 \leq 
\frac{2 |\tau|}{K} \, e^{|\tau|} \|f\|_2
\end{equation*}
for every $f \in L^2_{\mathbb{C}}({\mathbb{R}})$.
Finally, we note that 
\begin{align*}
& V \circ \left(\exp(K . {\textrm{id}}_{\mathbb{R}}) \cdot 
\chi_{_{[-|\tau|/K,|\tau|/K]}})\right) * V^{-1} g \\
& =
\exp(- K . {\textrm{id}}_{\mathbb{R}}) \left[
\left(\exp(K . {\textrm{id}}_{\mathbb{R}}) \cdot 
\chi_{_{[-|\tau|/K,|\tau|/K]}})\right) * 
(\exp(K . {\textrm{id}}_{\mathbb{R}}) \cdot g) \right] \\
& =
\chi_{_{[-|\tau|/K,|\tau|/K]}} * g
\end{align*}
for every $g \in L^2_{\mathbb{C}}({\mathbb{R}},e^{2 K \xi})$,
and collecting the obtained information we arrive 
at (\ref{representation3a}).

\end{proof}

\pagebreak

\end{document}